\documentclass[
reprint,
superscriptaddress,
amsmath,amssymb,
aps,
pra,
]{revtex4-2}

\usepackage{verbatim}
\usepackage{rotating, booktabs}  
\usepackage{subfigure} 
\usepackage{braket}
\usepackage{graphics}
\usepackage{graphicx}
\usepackage{braket}
\usepackage{listings}
\usepackage{color, xcolor}

\usepackage{upgreek}
\usepackage{dcolumn}
\newcolumntype{z}[1]{D{.}{.}{#1}}
\usepackage{float}
\definecolor{darkblue}{rgb}{0.0,0.0,0.3}
\usepackage[colorlinks=true,
            linkcolor=red,
            urlcolor= darkblue,
            citecolor=blue]{hyperref}

\usepackage{etoolbox}
\usepackage{url}
\usepackage{multirow} 
\usepackage[normalem]{ulem}

\usepackage[algoruled]{algorithm2e}

\usepackage{shadow,epsf,amsthm,amssymb,amsmath}

\usepackage{enumerate}

\usepackage{diagbox}
\usepackage{microtype}

\newtheorem{definitionenv}{Definition}
\newtheorem{lemmaenv}[definitionenv]{Lemma}
\newtheorem{theoremenv}[definitionenv]{Theorem}
\newtheorem{corollaryenv}[definitionenv]{Corollary}
\newtheorem{propositionenv}[definitionenv]{Proposition}
\newtheorem{conjectureenv}[definitionenv]{Conjecture}
\newtheorem{remarkenv}[definitionenv]{Remark}
\newenvironment{remark}{\begin{remarkenv}\rm}{\end{remarkenv}}
\newcommand{\br}{\begin{remark}}
	\newcommand{\er}{\end{remark}}

\newtheorem{exampleenv}{Example}
\newtheorem{app-lemmaenv}[section]{Lemma}

\newenvironment{definition}{\begin{definitionenv}\rm}{\end{definitionenv}}
\newenvironment{lemma}{\begin{lemmaenv}\rm}{\end{lemmaenv}}
\newenvironment{theorem}{\begin{theoremenv}\rm}{\end{theoremenv}}
\newenvironment{corollary}{\begin{corollaryenv}\rm}{\end{corollaryenv}}
\newenvironment{example}{\begin{exampleenv}\rm}{\end{exampleenv}}
\newenvironment{proposition}{\begin{propositionenv}\rm}{\end{propositionenv}}
\newenvironment{conjecture}{\begin{conjectureenv}\rm}{\end{conjectureenv}}
\newenvironment{app-lemma}{\begin{app-lemmaenv}\rm}{\end{app-lemmaenv}}

\newcommand{\bd}{\begin{definition}}
	\newcommand{\ed}{\end{definition}}
\newcommand{\bl}{\begin{lemma}}
	\newcommand{\el}{\end{lemma}}
\newcommand{\elp}{\hspace*{\fill} $\Box$
\end{lemma}}
\newcommand{\bt}{\begin{theorem}}
\newcommand{\et}{\end{theorem}}
\newcommand{\etp}{\hspace*{\fill} $\Box$
\end{theorem}}
\newcommand{\bc}{\begin{corollary}}
\newcommand{\ec}{\end{corollary}}
\newcommand{\ecp}{\hspace*{\fill} $\Box$
\end{corollary}}
\newcommand{\bcj}{\begin{conjecture}}
\newcommand{\ecj}{\end{conjecture}}

\newcommand{\be}{\begin{example}}
\newcommand{\ee}{\end{example}}
\newcommand{\eep}{\hspace*{\fill} $\Box$
\end{example}}
\newcommand{\bp}{\begin{proposition}}
\newcommand{\ep}{\end{proposition}}
\newcommand{\epp}{
\end{proposition}}

\usepackage{tikz}	
\usetikzlibrary{backgrounds,fit,arrows,decorations.pathreplacing,positioning}
\usepackage[colorinlistoftodos,prependcaption]{todonotes}
\usepackage{xargs}

\newcommandx{\yellownote}[2][1=]{\todo[inline,linecolor=yellow,backgroundcolor=yellow!25,bordercolor=yellow,#1]{#2}}

\newcommand{\wt}[1]{\mathrm{wt}\left(#1\right)}

\newcommand{\mbf}{\mathbf}

\newcommand{\cL}{{\cal L}}
\newcommand{\cM}{{\cal M}}

\begin{document}

 	\title{Enhancing Fault-Tolerant Surface Code Decoding with Iterative Lattice Reweighting}

 \author{Yi Tian}
 \thanks{These two authors contributed equally to this work.}
 \affiliation{Institute of Fundamental and Frontier Sciences,  \\ University of 
 Electronic Science and Technology of China, Chengdu    610051, China}

 \author{Y. Zheng}
 \thanks{These two authors contributed equally to this work.}
 \affiliation{Furlet Technology, Shanghai 200136, China }

 \author{Xiaoting Wang}
 \email{xiaoting@uestc.edu.cn}
 \affiliation{Institute of Fundamental and Frontier Sciences,  \\ University of 
 Electronic Science and Technology of China, Chengdu    610051, China}

 \author{Ching-Yi Lai}
\email{cylai@nycu.edu.tw}
 \affiliation{Institute of Communications Engineering, National Yang Ming Chiao Tung University, Hsinchu 30010, Taiwan}
 \date{\today}
 	
 \begin{abstract}
Efficient and realistic error decoding is crucial for viable fault-tolerant quantum computation (FTQC) on near-term quantum devices. While decoding is a classical post-processing task, its effectiveness relies heavily on accurately modeling quantum noise, which is inherently hardware-dependent.  Specifically, correlated bit-flip ($X$) and phase-flip ($Z$) errors frequently emerge under realistic circuit-level noise.
We introduce the Iterative Reweighting Minimum-Weight Perfect Matching (IRMWPM) decoder, which
systematically incorporates such correlations to enhance quantum error correction. 
Our approach  leverages fault-detection patterns to guide the reweighting: we identify correlated $X$ and $Z$ detection events and use their conditional probabilities to update weights on the primal and dual decoding lattices. 
This iterative procedure improves the decoder's ability to handle realistic error propagation in a hardware-agnostic yet noise-aware manner. 
We further prove that the IRMWPM decoder converges in finite time while preserving the distance guarantee of the original MWPM algorithm.
  Under circuit-level noise, our numerical results show substantial improvements in decoding performance.  For code distances $\geq 17$ and physical error rates $\leq 0.001$, IRMWPM reduces logical error rates by over 20× while maintaining computational efficiency with only a few iterations. It also raises the accuracy threshold from 1\% to 1.16\%, making it well-suited for near-term, real-time decoding in practical quantum computing architectures.
Extrapolated estimates based on our simulation fits suggest that to reach a logical error rate of $10^{-16}$, IRMWPM requires only distance $d=31$, whereas standard MWPM would require $d=50$, implying a substantial reduction in qubit overhead for achieving high-accuracy FTQC. These results make our method well-suited for near-term, real-time decoding in practical quantum computing architectures.

  	\end{abstract}

 	\maketitle

 	\section{Introduction}
 	Quantum noise presents a significant challenge to the realization of reliable quantum computation. To overcome this, the development of quantum error-correcting codes (QECCs) and fault-tolerant quantum computation (FTQC) is essential~\cite{aharonov1997fault, gottesman1997stabilizer, nielsen2010quantum, lidar2013quantum}. It has been rigorously established that large-scale quantum computation is achievable if the physical noise remains sufficiently local and below a critical threshold~\cite{aharonov1997fault, gottesman1997stabilizer, aliferis2005quantum, campbell2017roads}.

     Recent advances in hardware platforms---such as superconducting qubits~\cite{krinner2022realizing, zhao2022realization, AAA+23,AAA+24,lacroix2025scaling}, trapped ions~\cite{egan2021fault, ryan2021realization,HPS+22}, and neutral atom atom systems~\cite{bluvstein2022quantum, ebadi2021quantum,sales2025experimental,bluvstein2025architectural}---have enabled devices with dozens of qubits and demonstrated the key control and fidelity necessary for proof-of-principle demonstrations of fault-tolerant quantum error correction (FTQEC). Among QECCs, surface codes have emerged as leading candidates due to their local stabilizer structure, high thresholds, and amenability to scalable hardware~\cite{kitaev2003fault, bravyi1998quantum, dennis2002topological, fowler2012surface}.

     In FTQEC, syndrome extraction produces a continuous stream of data that must be decoded promptly to identify the locations and types of errors. An effective decoder can achieve code-capacity performance, enabling smaller code distances and reducing overall resource overhead. To support real-time processing and prevent backlog accumulation, decoding must keep pace with syndrome extraction rates---particularly before executing non-Clifford operations~\cite{terhal2015quantum, skoric2023parallel,tan2023scalable,caune2024demonstrating}. This necessitates a careful balance between decoder complexity and performance, which is critical for scalable  FTQC.

 	A wide range of decoders have been proposed for surface codes, including Minimum Weight Perfect Matching (MWPM)\cite{dennis2002topological, wang2011surface, fowler2013minimum, wu2023fusion, higgott2025sparse}, Union-Find (UF)~\cite{delfosse2021almost, huang2020fault}, renormalization group (RG)~\cite{DP10,duclos2010fast}, belief propagation (BP)~\cite{KL25,KL24b}, Markov chain Monte Carlo (MCMC)~\cite{hutter2014efficient}, tensor networks (TN)~\cite{darmawan2017tensor, darmawan2018linear}, and neural networks (NN)~\cite{varsamopoulos2017decoding, ni2020neural, chamberland2018deep, sweke2020reinforcement, zhang2023scalable,  chamberland2023techniques,bausch2024learning}. Each approach comes with trade-offs in accuracy, complexity, and scalability~\cite{demarti2024decoding}.

    For example, RG and BP decoders achieve $O(L^2 \log L)$ complexity for distance-$L$ codes, but both methods lack theoretical guarantees for their decoding performance at low error rates~\cite{duclos2010fast,KL25}. MCMC methods yield improved thresholds but incur higher computational costs. TN methods are effective under general noise models, but their performance under noisy syndrome extraction remains unclear. NN decoders provide fast inference after training but struggle to generalize across code sizes.

    The UF decoder offers linear-time complexity and favorable scalability for near-term deployment~\cite{delfosse2021almost, huang2020fault, das2022afs, liyanage2023scalable,barber2025real}, while MWPM remains the most widely used high-performance decoder. MWPM achieves a threshold near the optimal under circuit-level noise but exhibits high computational complexity, particularly in the presence of measurement errors. Improvements such as Sparse Blossom and Parity Blossom reduce its average-case runtime to $O(L^3)$~\cite{higgott2025sparse, fowler2013minimum, wu2023fusion}, greatly improving its practical utility on devices like 
    Field-Programmable Gate Array
    (FPGA)~\cite{micro_blossom}. 
 
    One limitation of the standard MWPM decoder is its independent treatment of $X$ and $Z$ errors, neglecting correlations that arise in noise models such as depolarizing noise.  For instance, the presence of a $Z$ error implies a 50\% chance of a co-occurring $X$ error on the same qubit, motivating the use of correlated decoding strategies~\cite{Fow13, DT14}. Note that such $X$-$Z$ correlations are prevalent in realistic noise channels of quantum memories and gates beyond depolarization noise.

We study an iterative MWPM scheme that reweights the $X$ and $Z$ decoding graphs using the most recent estimates from the dual lattice. This idea is not new: Fowler explicitly described a single reweighting step under the code-capacity noise model~\cite{Fow13}, which inspired subsequent work on iterative reweighting~\cite{yuan2022modified,iOMFC23}. However, prior approaches capture only simple correlations and do not analyze the richer 
$X/Z$ correlations induced by two-qubit CNOT errors—a dominant factor in the circuit-level noise model with noisy syndrome-extraction circuits.
Moreover, it remains unclear whether such iterative strategies preserve the decoding radius of standard MWPM. For example, recent work has shown that iterative versions of the UF decoder can violate the distance guarantees of the original UF decoder~\cite{LL25}.

In this paper, we propose a fault-tolerant iterative reweighting MWPM (IRMWPM) decoder that systematically exploits 
$X/Z$ error correlations by leveraging known fault-path structures in noisy syndrome measurements~\cite{wang2011surface}. By identifying correlated detection event patterns and calculating their conditional probabilities (Table~\ref{tb:reweighting}), we develop a principled, data-driven reweighting strategy. Our IRMWPM decoder alternates weight updates between the 3D primal and dual decoding lattices, enhancing its ability to capture realistic error propagation. We prove that IRMWPM converges in finite time, with empirical results showing that only two to four iterations are typically sufficient at practical error rates. Finally, we prove that IRMWPM preserves the decoding radius of the original MWPM decoder.

\section{Results}\label{sec:sim}
We propose the IRMWPM decoder, which exploits $X$--$Z$ correlations arising from circuit-level noise. In this section, we illustrate the decoding procedure, present its theoretical guarantees, and demonstrate its performance through numerical simulations. The complete decoding algorithm is given in Algorithm~\ref{alg:iterative_mwpm}.

\subsection{Iterative reweighting decoding strategy}

\label{sec:iterative}

Let $\mathcal{L}_X^{(j)}$ and $\mathcal{L}_Z^{(j)}$ denote the decoding lattices for $X$-type and $Z$-type errors, respectively, after $j$ rounds of reweighting for a given error syndrome.
Let $\mathcal{M}_X^{(j)}$ and $\mathcal{M}_Z^{(j)}$ be the corresponding outputs of the MWPM algorithm applied to $\mathcal{L}_X^{(j)}$ and $\mathcal{L}_Z^{(j)}$.
In particular, $\mathcal{L}_X^{(0)}$ and $\mathcal{L}_Z^{(0)}$ are the initial decoding lattices from given error syndromes $S_X$ and $S_Z$ for the $X$ and $Z$ errors, and $\mathcal{M}_X^{(0)}$ and $\mathcal{M}_Z^{(0)}$ are the initial correction outputs from the standard MWPM decoder.

The iterative decoding process proceeds as follows.
First, use $\mathcal{M}_X^{(0)}$ to reweight $\mathcal{L}_Z^{(0)}$ and obtain an updated lattice $\mathcal{L}_Z^{(1)}$.
Next, apply the MWPM algorithm to $\mathcal{L}_Z^{(1)}$ to obtain a new matching $\mathcal{M}_Z^{(1)}$.
Then, use $\mathcal{M}_Z^{(1)}$ to reweight $\mathcal{L}_X^{(0)}$ and obtain $\mathcal{L}_X^{(1)}$, followed by applying MWPM to get $\mathcal{M}_X^{(1)}$. 
This sequence constitutes one full iteration. The process is repeated, alternating between reweighting and decoding $X$-type and $Z$-type lattices,   until the decoding results converge or a predefined iteration limit is reached.
Figure~\ref{fig:IRMWPM} illustrates the IRMWPM decoding framework. 
\begin{figure}[htp]
	 \includegraphics[width=0.99\columnwidth]{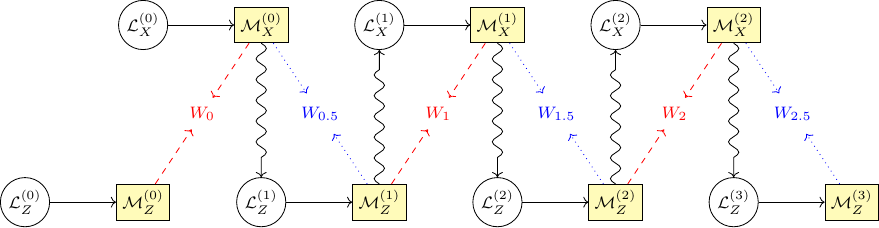}
	\caption{Illustration of the iterative   decoding framework.}
	\label{fig:IRMWPM}
\end{figure}

\subsubsection{Reweighting strategy}

The IRMWPM decoder differs from the conventional MWPM decoder by dynamically updating edge weights on the decoding lattice  based on dual-lattice matching outcomes, whereas MWPM uses fixed weights throughout decoding.

\begin{figure}[htbp]
	\centering
	\subfigure[]{\includegraphics[width=0.45\columnwidth]{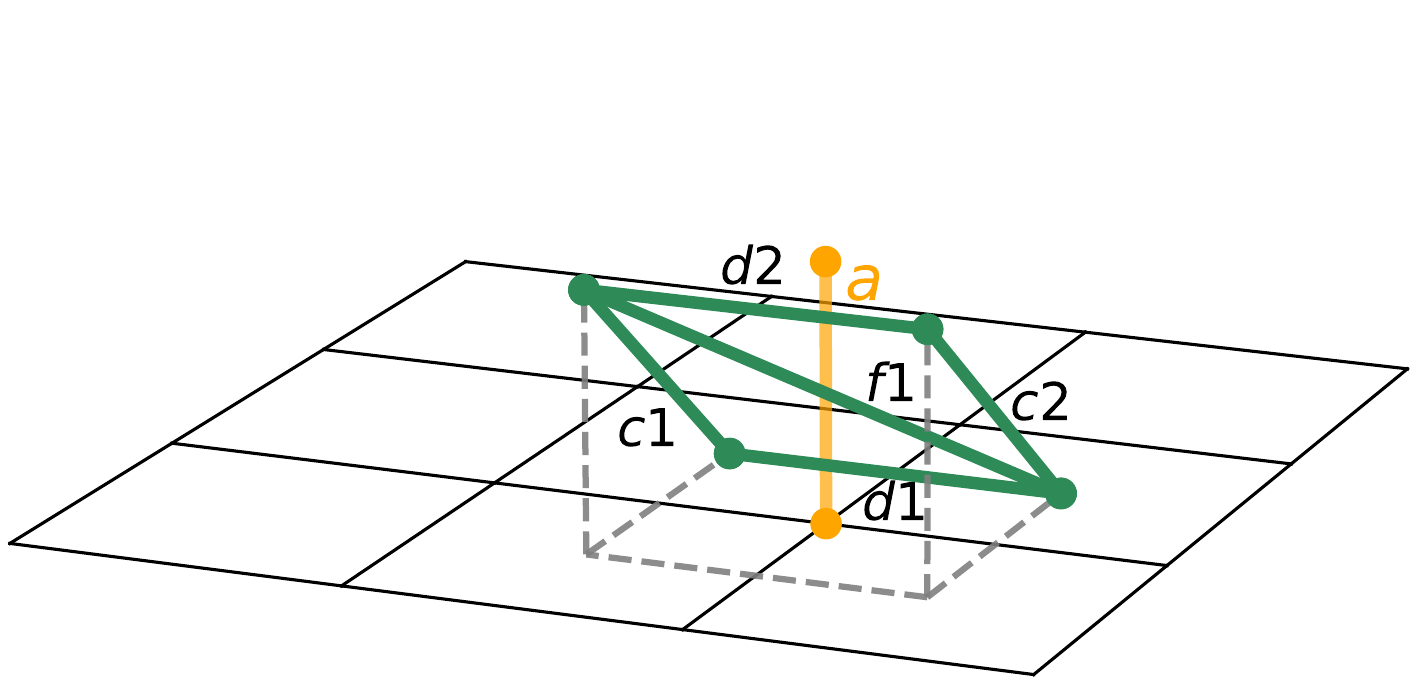}} \hfill
	\subfigure[]{\includegraphics[width=0.45\columnwidth]{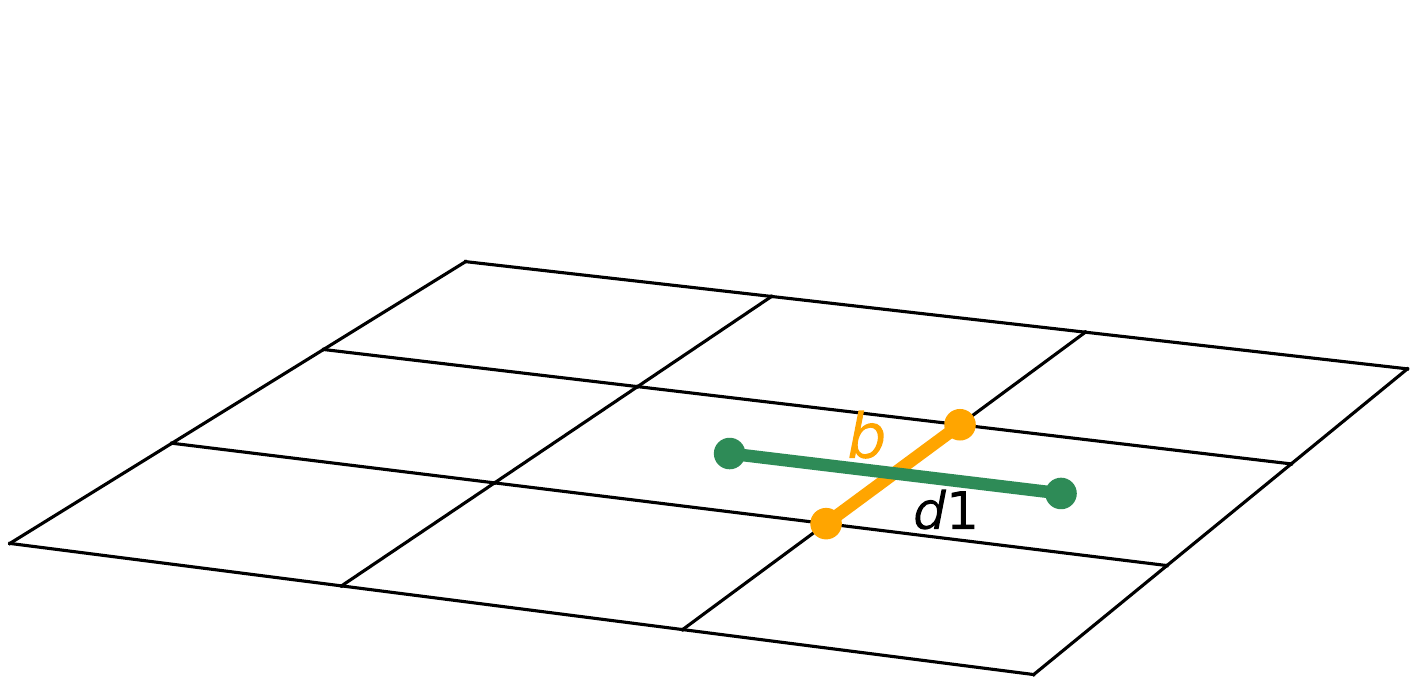}} \\[1ex]
	\subfigure[]{\includegraphics[width=0.45\columnwidth]{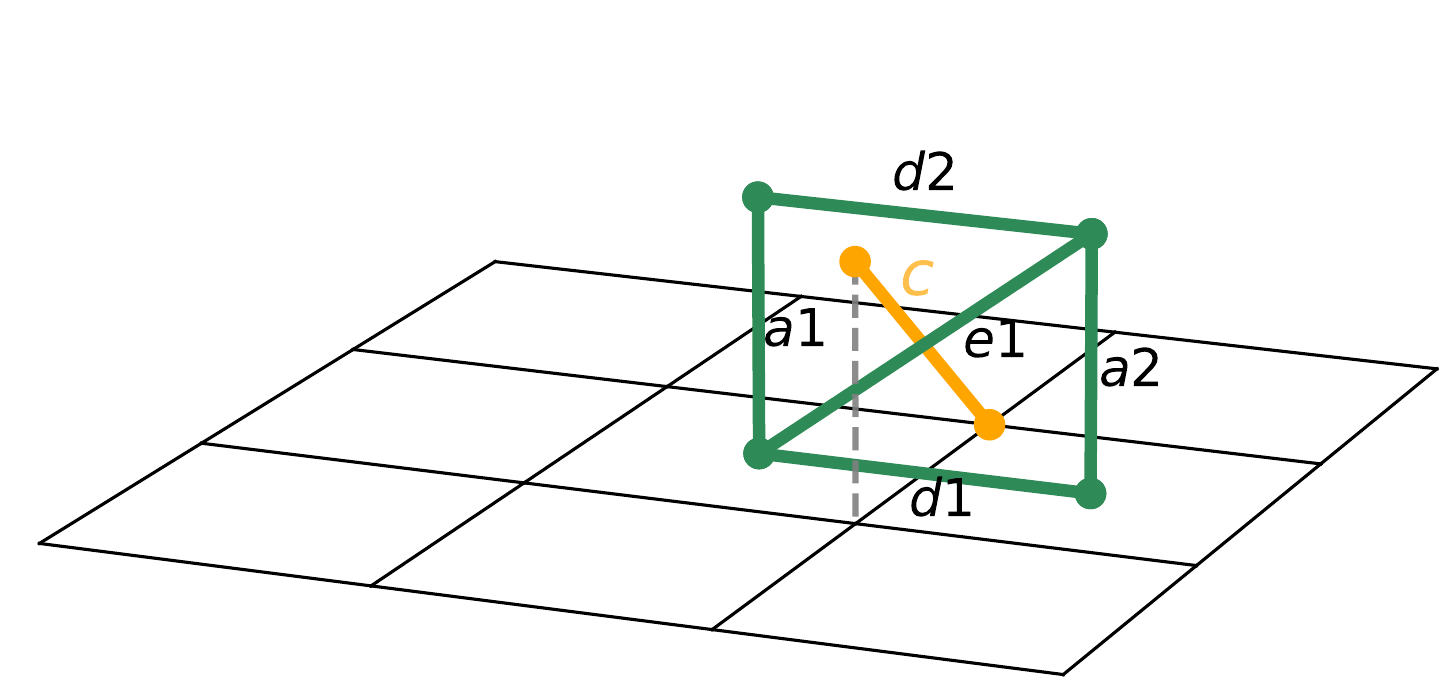}} \hfill
	\subfigure[]{\includegraphics[width=0.45\columnwidth]{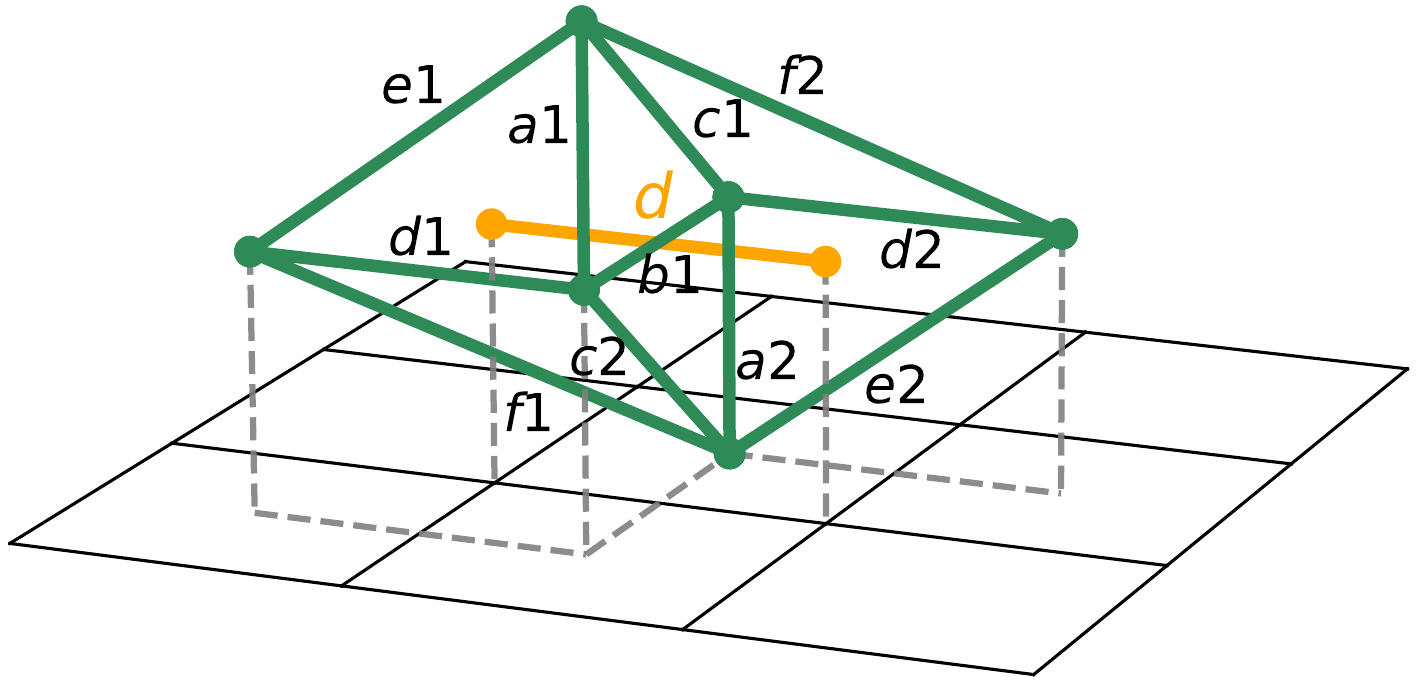}} \\[1ex]
	\subfigure[]{\includegraphics[width=0.45\columnwidth]{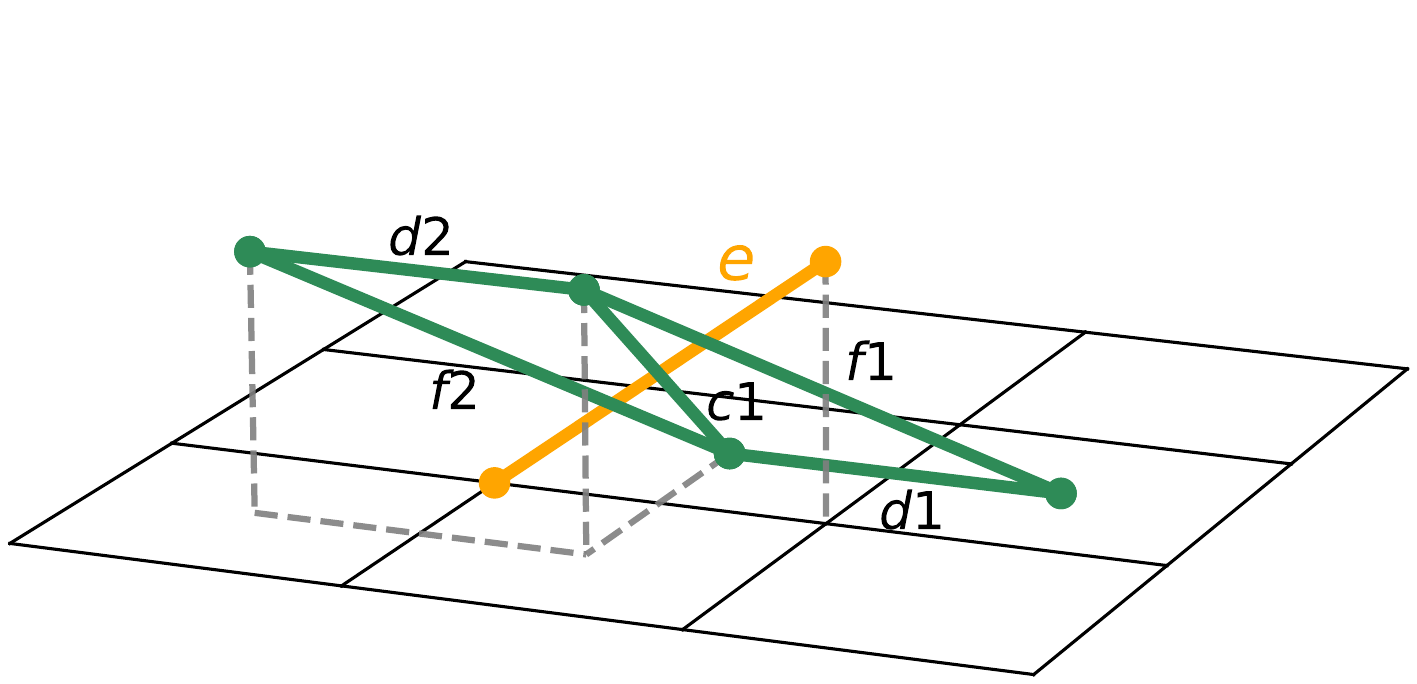}} \hfill
	\subfigure[]{\includegraphics[width=0.45\columnwidth]{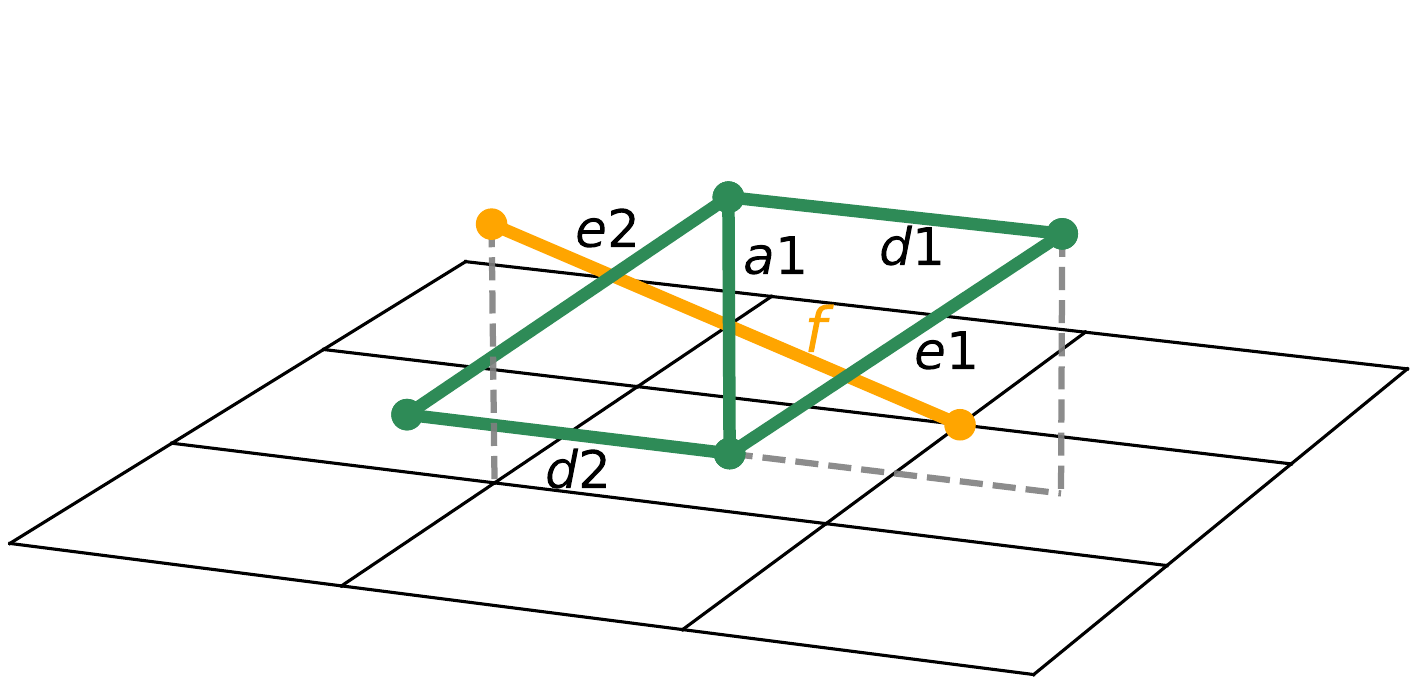}} 
	\caption{Decoding lattices  $\mathcal{L}_X$ for $L=3$. When there is an $X$ error matching (orange edge), the corresponding $Z$ error matching (green edge) in $\mathcal{L}_Z$ needs to be reweighted.}
	\label{fig5}
\end{figure}

The $X$ and $Z$ decoding lattices are dual to each other.
In the iterative reweighting procedure,  the edge weights in one lattice are updated based on the decoding outcomes of its dual.
To exploit realistic fault correlations in circuit-level noise, we incorporate information about fault-induced correlations between detection events on the $X$ and $Z$ decoding lattices, as identified in Wang et al.~\cite{wang2011surface}. Specifically, we classify and group detection event matchings that arise from single faults, such as CNOT errors, data qubit errors, or measurement faults, as shown in Fig.~\ref{fig5}.

Each single fault may produce simultaneous detection events on both the $X$ and $Z$ decoding lattices. For instance, a single $Y$ error on a data qubit can trigger correlated detection events corresponding to matching edges in both lattices. These correlated edge patterns are systematically categorized into six types (a–f), each of which appears as a characteristic “thick edge” in the decoding graph. The resulting correlated pairs of detection events are matched across primal and dual lattices, forming fault-path structures that can be used to compute conditional probabilities of dual lattice matchings given observed primal events.

Using this classification, we identify detection event pairs that frequently co-occur due to a common fault source and use their empirical joint and conditional probabilities to reweight edges in the dual decoding lattice. This principled, data-driven approach allows us to iteratively update the edge weights between the primal and dual lattices during decoding, improving the MWPM decoder’s ability to capture realistic error propagation patterns.
These conditional probabilities are summarized in Table~\ref{tb:reweighting} (Section~\ref{sec:method}).

\subsubsection{Theoretical Guarantees of IRMWPM}

Define $\hat{E}^{(j)}$ as the Pauli correction associated with $\cM_X^{(j)}$ and $\cM_Z^{(j)}$, and let $\hat{E}^{(j+0.5)}$ denote the correction associated with $\cM_X^{(j)}$ and $\cM_Z^{(j+1)}$.
For each integer $j \geq 0$, define $W_j = \wt{\hat{E}^{(j)}}$ and $W_{j+0.5} = \wt{\hat{E}^{(j+0.5)}}$ as the total weights of the corresponding corrections. 
 These quantities are also illustrated in Fig.~\ref{fig:IRMWPM}.
 
We prove in Lemma~\ref{lemma:nonincreasing} (Section~\ref{sec:method}) that the sequence ${W_j}$ is nonincreasing and stabilizes in finite time. This leads to the following result:

\begin{theorem} \label{thm:convergence}
The IRMWPM decoder converges in finite time.
\end{theorem}

In practice, our simulations show rapid convergence, typically within 2–4 iterations for realistic error rates in the range of $10^{-3}$ to $10^{-4}$.

We prove that the IRMWPM decoder achieves the same decoding radius as the standard MWPM decoder, thereby preserving its distance guarantee.
\begin{theorem}[IRMWPM]
Let $E$ be an $n$-qubit Pauli error with weight at most $ \left\lfloor \frac{d-1}{2} \right\rfloor$. 
If $E$ is correctable by the MWPM decoder on a surface code of distance~$d$ under either the code capacity or circuit-level noise model, 
then $E$ is also correctable by the IRMWPM decoder, 
assuming that the final round of syndrome measurements is perfect.
\end{theorem}

 \subsection{Numerical Simulations}

We numerically evaluate the performance of the proposed IRMWPM decoder under the standard symmetric circuit-level depolarizing noise model, focusing on its impact on logical error rates, threshold behavior, and decoding efficiency.

The simulation proceeds as follows: 
For a distance-$L$ surface code, we assume noisy syndrome extraction (SE) with circuit-level noise in each round. After performing $L$ rounds of SE, a decoder is applied using the accumulated syndromes.
 In addition, every $T$ SE cycles, a virtual decoding step is performed using an ideal decoder that assumes one round of perfect measurements and operates on a 2D lattice to check for logical errors. If no logical error is detected, noisy SE continues.
This process repeats until a logical error occurs, and we record the total number of SEs. We repeat this procedure to compute the average logical qubit lifetime. In the following simulations, we set $T=L$.
Similar lifetime simulations were conducted in~\cite{KL24b} using BP decoding. We show that IRMWPM achieves better performance.
 
In a single run of IRMWPM, iterations stop once the decoding result repeats (i.e., the same result is obtained in two consecutive iterations), which typically occurs within 2–4 iterations in our simulations. Correction weight stabilization can also serve as a secondary stopping criterion if perfect convergence is desired.

Benchmark simulations are conducted using surface codes of distance $L = 5, 7, 9, 13, 17$, with a decoding period of $T = L$ rounds of syndrome extraction.

\begin{figure}[htbp]
	\centering
	\subfigure[]{\includegraphics[width=1\columnwidth]{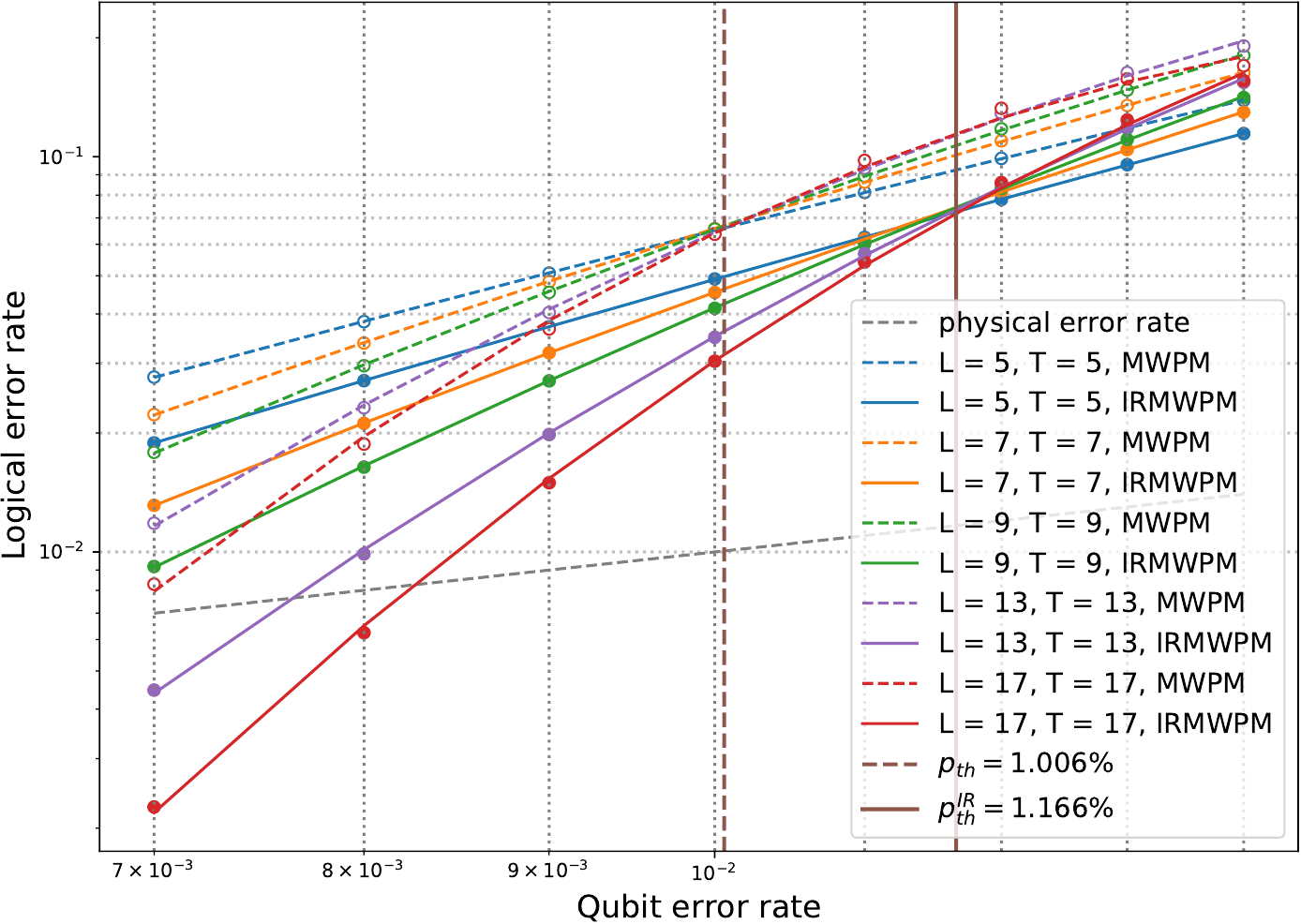}}
	\caption{Comparison of decoding thresholds achieved by the MWPM and IRMWPM decoders.}
	\label{mwpm_threshold}
\end{figure}

The decoding threshold denotes the physical error rate below which the logical error rate can be suppressed arbitrarily by increasing the surface code distance. As shown in Fig.~\ref{mwpm_threshold}, the standard MWPM decoder yields a threshold of approximately 1\%, while the  IRMWPM decoder achieves a higher threshold of about 1.16\%, representing a 16\% improvement.

\begin{figure}[htbp]
	\centering
	\includegraphics[width=1\columnwidth]{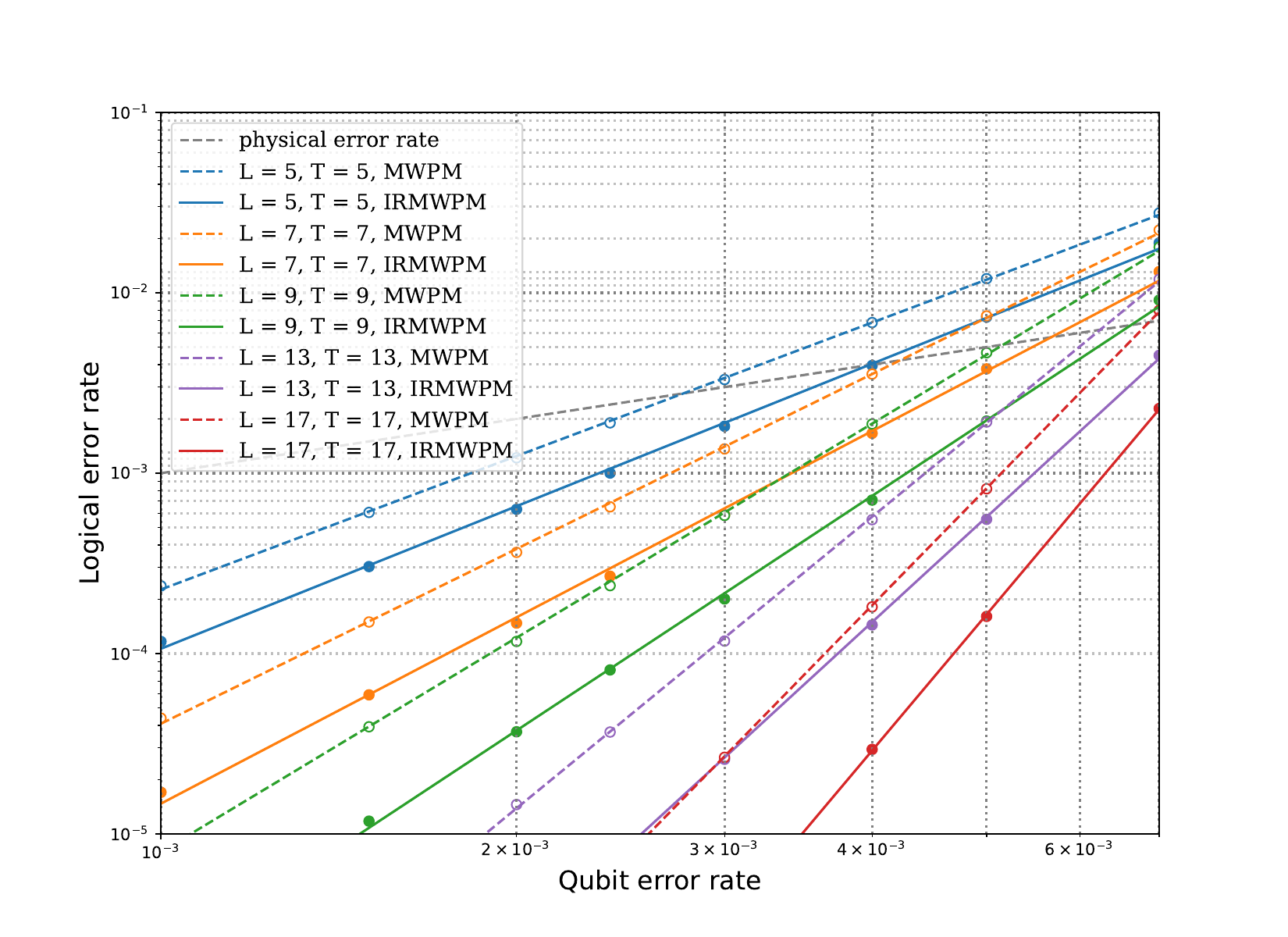}
	\caption{Decoding performance of MWPM and IRMWPM decoders at low physical error rates.}
	\label{mwpm_rate}
\end{figure}

\begin{figure}[htbp]
	\centering
	\includegraphics[width=0.85\columnwidth]{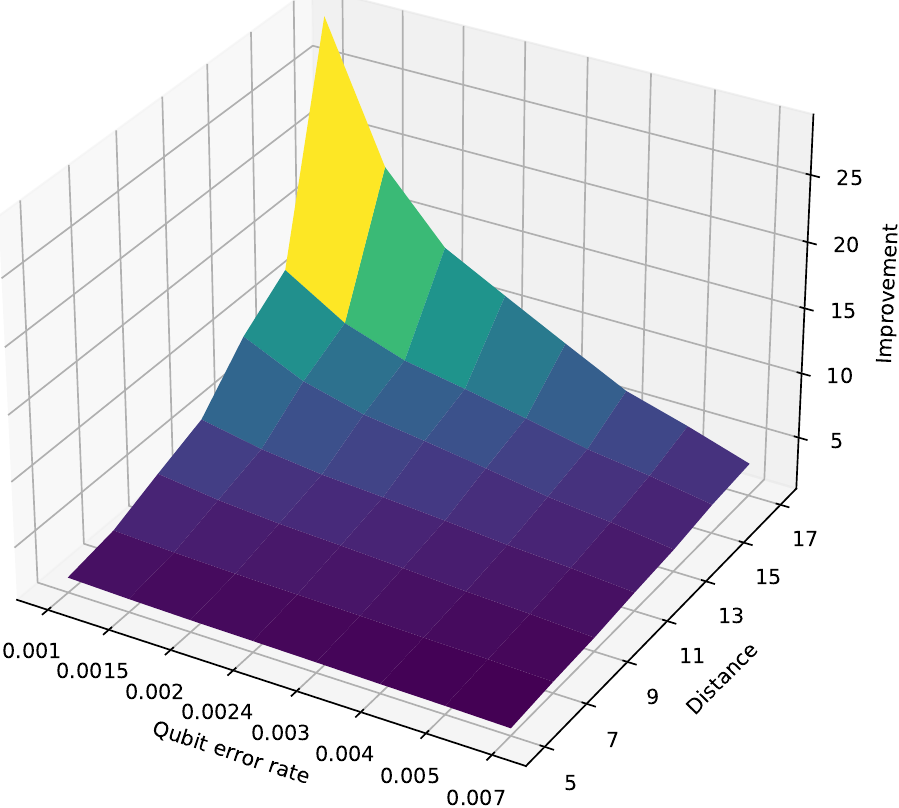}
	\caption{Impact of iterative lattice reweighting on decoding performance across different physical error rates and surface code distances.}
	\label{mwpm_improvement}
\end{figure}

\begin{figure}[htbp]
	\centering
	\subfigure[]{\includegraphics[width=0.9\columnwidth]{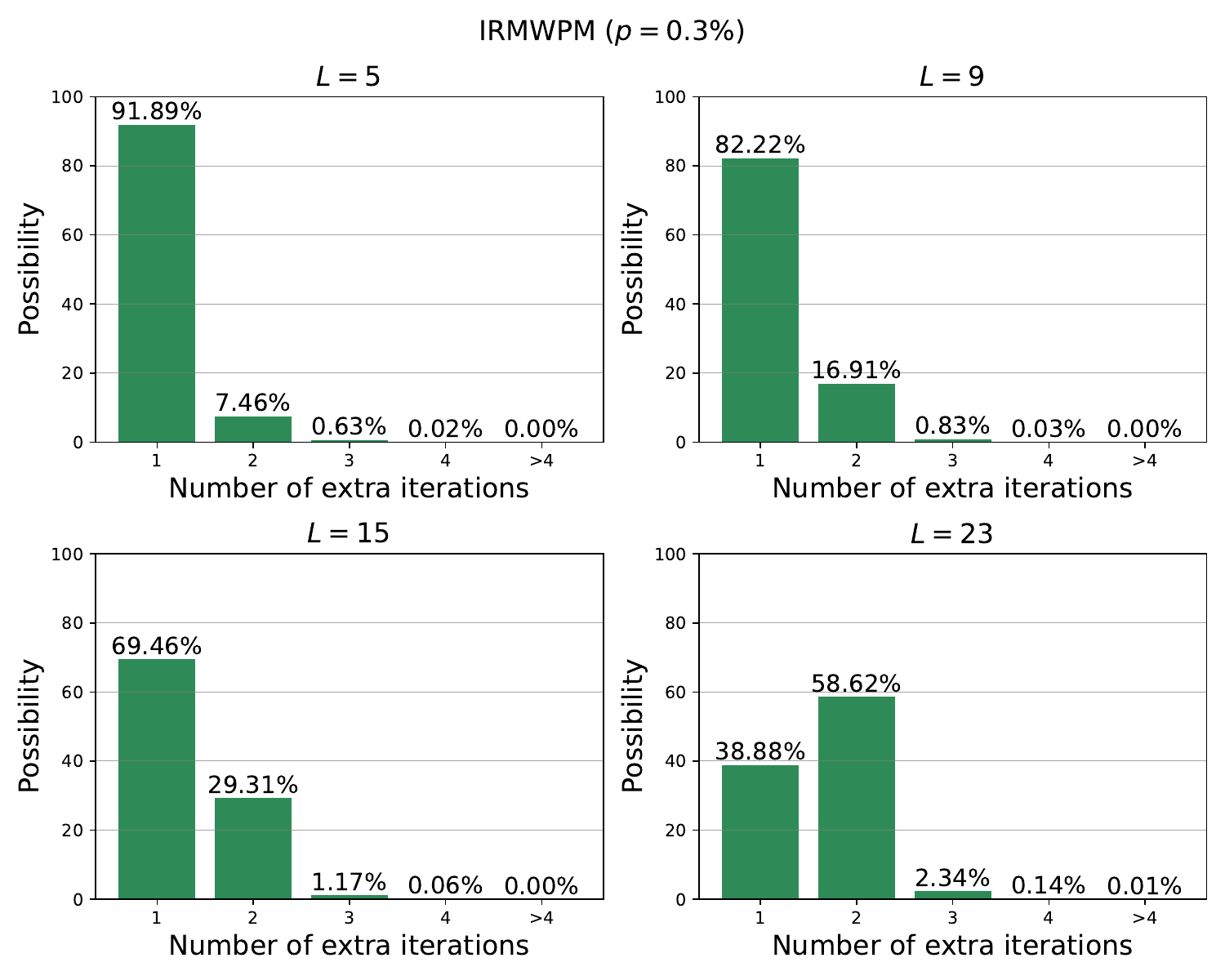}}\\[1ex]
	\subfigure[]{\includegraphics[width=0.9\columnwidth]{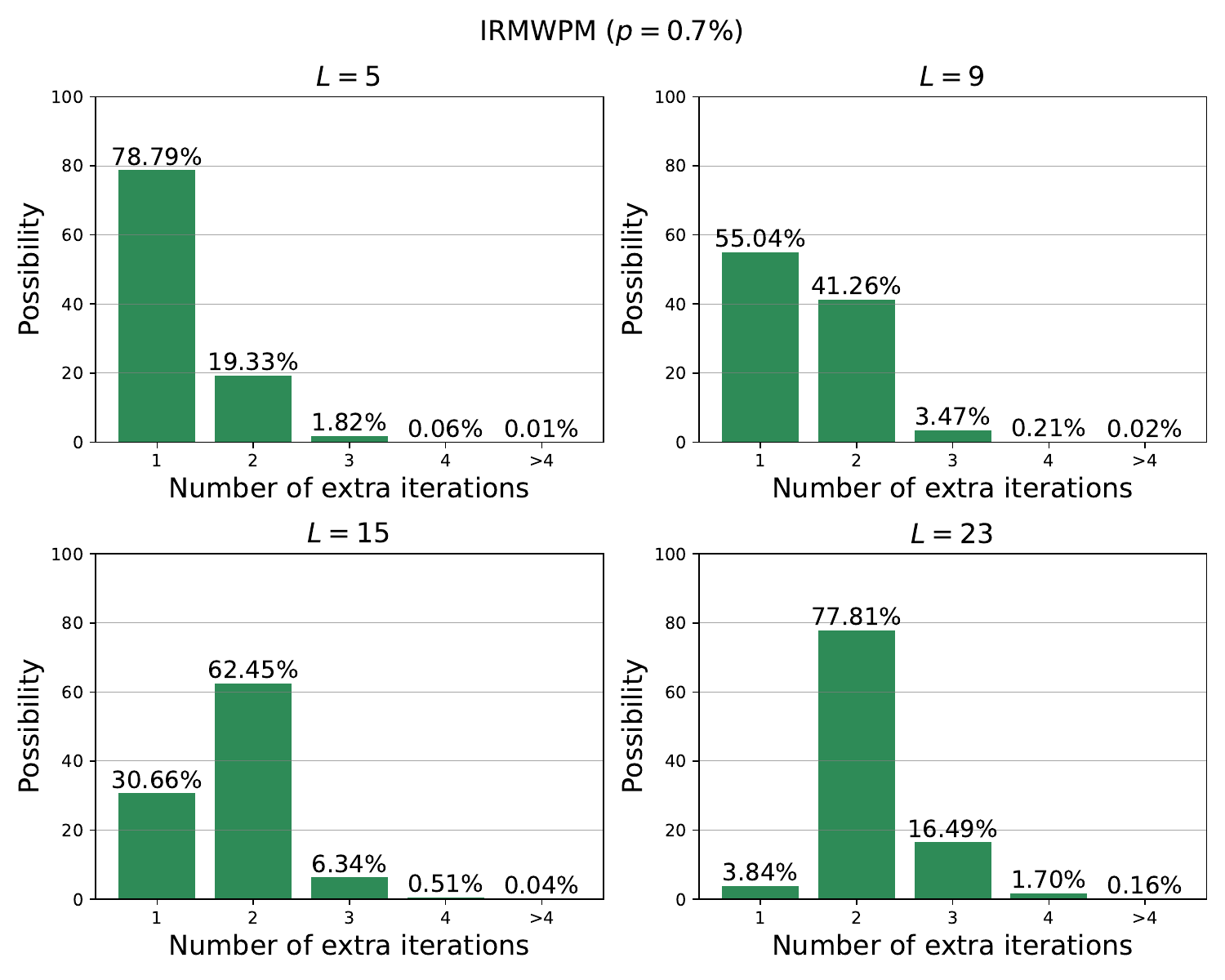}}
	\caption{Number of additional iterations required for convergence under various physical error rates and surface code distances in the IRMWPM decoder. }
	\label{iter}
\end{figure}

Next, we examine the impact of iterative lattice reweighting on the decoding error rate of MWPM decoders, particularly in the low physical error rate regime. As illustrated in Fig.~\ref{mwpm_rate}, iterative reweighting significantly reduces the decoding error rate. For example, at a physical error rate of $p = 0.24\%$ and a surface code distance of $L = 13$, the IRMWPM decoder achieves an approximately 80\% reduction in decoding error rate compared to the standard MWPM decoder.

Here, we quantitatively evaluate the improvement in the performance of the MWPM decoder using iterative lattice reweighting across physical error rates $p$ and surface code distances $L$. As shown in Fig.~\ref{mwpm_improvement}, 
the advantage of this iterative approach becomes more significant at lower $p$ and larger $L$, 
leading to a more substantial reduction in decoding error rates.
To estimate the logical error rates for both MWPM and IRMWPM decoders, we extrapolate from the data in Fig.~\ref{mwpm_rate}, which is fitted using the following general equation:
\begin{equation}
 P_\textrm{L}(p,L) = 10^{aL^2+bL+c}p^{eL^2+fL+g},   
\end{equation}
where $a, b, c, e, f,$ and $ g$ are fitted parameters. Extrapolating to $L = 31$ and $p = 0.001$ (a typical regime for large-scale FTQC running Shor’s algorithm), we find $P_\textrm{L} \approx 3\times10^{-12} $ for MWPM and $3\times10^{-16}$ for IRMWPM, respectively, showing a reduction of approximately four orders of magnitude. To match the same logical error rate using standard MWPM, a surface code of distance $L\geq50$ is needed. Thus, IRMWPM offers a physical qubit saving of over 60\% per logical qubit.

 While iterative reweighting enhances decoding performance, it introduces additional computational cost due to multiple iterations. Assessing convergence requirements is thus critical: as shown in Fig.~\ref{iter}, the IRMWPM 
 decoder averages only a few extra iterations, 
 
 For $L=23$, four iterations are sufficient for $99\%+$ of cases at $p=0.3\%$.
 When operating at the lower physical error rates,   this average would be significantly less, 
keeping the overall computational complexity comparable to standard MWPM.

\section{Discussion}

This work presents the IRMWPM decoder as a principled and scalable enhancement MWPM decoding under circuit-level noise. By systematically incorporating known fault-path correlations into an iterative lattice reweighting process, IRMWPM improves decoding accuracy while preserving the distance guarantee of the original decoder.
In contrast to previous reweighting strategies, which either lack theoretical justification or are limited to code capacity models, our approach is grounded in well-characterized detection event patterns and fault mechanisms, enabling interpretable and data-driven reweighting on the 3D decoding lattice. We alternately update weights between the primal and dual grids, allowing the decoder to capture realistic error propagation better.

Numerical simulations demonstrate a 16\% increase in accuracy threshold and up to an 80\% reduction in logical error rate at low physical error rates, all while requiring only a few additional iterations. This efficiency makes IRMWPM particularly suitable for near-term, real-time deployment in fault-tolerant quantum architectures.


IRMWPM converges quickly in practice, as shown by our simulations. We formally prove its finite-time convergence by demonstrating that the total weight of the $X$ and $Z$ error estimates constitutes a non-increasing sequence of non-negative integers. This not only guarantees convergence but also yields a straightforward, weight-based termination criterion.  This leads to an open theoretical challenge: proving IRMWPM's convergence within $O(\textrm{poly}(\log L))$.

These results highlight the practical value of leveraging correlated errors in decoding, without sacrificing theoretical robustness. Future directions include extending this framework to other topological codes, exploring approximate fault models for broader applicability. Given the significant progress in deploying MWPM on FPGAs to achieve microsecond-level decoding latency for code length \( L = 13 \)~\cite{micro_blossom}, IRMWPM, which incurs lower overhead than MWPM, holds promise for high-performance real-time quantum error correction when implemented on FPGAs or ASICs.

The software developed in this paper is available at Github~\cite{FTIR25}.

\section{Method} 	
\label{sec:method}

\subsection{Preliminaries} \label{sec:preliminaries}

\subsubsection{Pauli operators}
We consider errors that are tensor products of Pauli operators $I,X,Y,Z$.
A single-qubit Pauli operator can be expressed as $X^a Z^b$ for $a, b \in \{0,1\}$, up to  a global phase, where $X^0=I=Z^0$, $X^1=X$, and $Z^1=Z$.
In general, an $n$-fold Pauli operator can be  expressed as $X^{\mbf{a}} Z^{\mbf{b}}\triangleq \bigotimes_{j=1}^n X^{a_j}Z^{b_j}$,
where $\mbf{a}= a_1\cdots a_n $ and $\mbf{b}= b_1\cdots b_n $, up to a global phase.
For simplicity, a Pauli operator $E$ is written as $E = E_X E_Z$, where $E_X$ and $E_Z$ denote its $X$ and $Z$ components, respectively.

The support of a Pauli error $E$ is the set of qubits where $E$ acts nontrivially.
Let $\wt{E}$ denote the weight of a Pauli operator $E$, defined as the number of its nontrivial components, which is the size of its support.  
Additionally,  we use $\wt{\mbf{a}}$ to denote the Hamming weight of the binary vector $\mbf{a}$.

\subsubsection{Surface code}
 
 The surface code is a quantum error-correcting code implemented on a two-dimensional lattice of physical qubits.
 We focus on the $[[L^2 + (L - 1)^2, 1, L]]$ surface code~\cite{bravyi1998quantum}, which encodes a single logical qubit into $L^2 + (L - 1)^2$ physical qubits arranged on a lattice of length $L$, as illustrated in Fig.~\ref{fig1} for the case of $L = 3$.
 
 \begin{figure}[htbp] 
 	\centering 
 	\includegraphics[width=0.5\columnwidth]{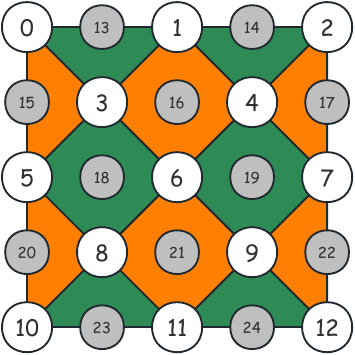}
 	\caption{Implementation of a distance-3 surface code on a two-dimensional lattice, where gray circles denote measurement qubits and white circles represent data qubits.}
 	\label{fig1}
 \end{figure}

In this layout, physical qubits are categorized into two types: measurement qubits, shown as gray circles in Fig.~\ref{fig1}, and data qubits, represented by white circles.
Each data qubit participates in at least two stabilizer generators, indicated by the orange and green color blocks. Each block corresponds to a measurement qubit and its adjacent data qubits, forming a Pauli stabilizer of weight 3 or 4.

The parameter $L$ also defines the minimum distance of the surface code, which is the minimum weight of a Pauli error that can cause a logical failure.
For example, a horizontal $X$ string acting on qubits 10, 11, and 12 constitutes a logical $X$ error.

An $X$ stabilizer   is defined by the tensor product of $X$ operators acting on the neighboring data qubits of a measurement qubit in an orange block, with $Z$ stabilizers defined similarly. Measuring these stabilizers determines the parity of the corresponding data qubits in the $Z$ or $X$ basis, as shown in Fig.~\ref{fig2}. The measurement outcomes, known as \textit{error syndromes}, help identify both the type and location of errors on the data qubits. Each stabilizer measurement involves initializing a measurement qubit, applying four CNOT gates with adjacent data qubits, and performing a final measurement. 
This process is also referred to as \textit{syndrome extraction} (SE).

For more efficient practical simulations, we adopt a reduced-depth implementation of the SE circuit consisting of five steps,  by removing the $I$ and $H$ gates and instead preparing ancilla qubits in the $\ket{0}$ and $\ket{+}$ states, with measurements performed in the $Z$ and $X$ bases.

\begin{figure}[htbp] 
	\centering 
	\includegraphics[width=0.28\columnwidth]{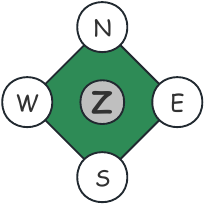}
	\includegraphics[width=0.6\columnwidth]{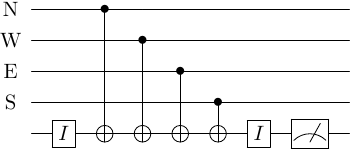}
	\includegraphics[width=0.28\columnwidth]{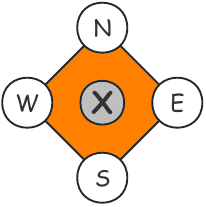}
	\includegraphics[width=0.6\columnwidth]{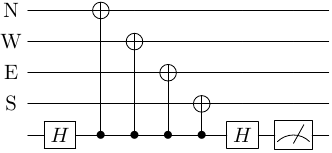}
	\caption{Standard $Z$- and $X$-stabilizer syndrome extraction circuits for the surface code. N, W, E, and S denote the data qubits located to the north, west, east, and south of the measurement qubits, respectively.
		Weight-three stabilizers can be measured in a similar manner.}
	\label{fig2}
\end{figure}

\subsubsection{Noise model}\label{sec:noise_model}

We define the terminology and noise model as follows.

A fault refers to a malfunction or imperfection in a physical operation or process, while an error is the unintended deviation in the quantum state resulting from such a fault. For example, a single-qubit gate fault  may introduce a Pauli error on the affected qubit.

In this work, we adopt a standard circuit-level depolarizing noise model for gate and measurement faults in SE, where $p$ denotes the fault rate~\cite{wang2011surface}. Our syndrome extraction circuits include the following three types of faults:

\begin{enumerate}[(1)]
	\item %
	Following the execution of a single-qubit identity $I$ or the Hadamard gate $H$, a Pauli error from  $\{X, Y, Z\}$  is applied to the data qubit with probability $p/3$.

	\item After each two-qubit CNOT gate, a two-qubit Pauli error is uniformly and randomly drawn from the set $\{I, X, Y, Z\}^{\otimes2}\setminus\{I\otimes I\}$ with probability  $p/15$.
	
	\item Performing a two-outcome measurement 
	produces an incorrect  outcome with probability $p$. 
\end{enumerate}

A round of SE for all $X$-type and $Z$-type stabilizers can be performed in parallel using the circuits shown in Fig.~\ref{fig2}.
However, due to imperfections in the involved gates and measurements, error syndromes may be incorrect.
Decoding based on such faulty syndromes can lead to a high probability of logical errors.
For a surface code of distance $L$, it generally requires decoding over $T$ rounds of SE to ensure fault tolerance, where $T = O(L)~$\cite{dennis2002topological}.
In this paper, we set $T = L$.

\subsubsection{ Decoding lattice and the MWPM decoder}

In the circuit-level noise model, the decoding lattice for surface codes is a 3D space-time graph that accounts for possible faults across multiple rounds of syndrome extraction~\cite{wang2011surface,fowler2012topological}. 
Two separate decoding lattices are used for decoding $X$-type and $Z$-type errors, respectively. 
Each stabilizer measurement at different time steps is represented by a check node. Within a single round, two check nodes are connected if they involve the same qubits. 
For example, measurement errors induce vertical edges connecting the same stabilizer across two consecutive rounds. Additionally, CNOT-induced faults create characteristic edge patterns between check nodes across time steps, capturing the correlated nature of propagated errors.
Figure~\ref{fig3} illustrates the decoding lattice for $Z$-type stabilizers with $L=3$.

\begin{figure}[htbp] 
	\centering 
	\includegraphics[width=0.7\columnwidth]{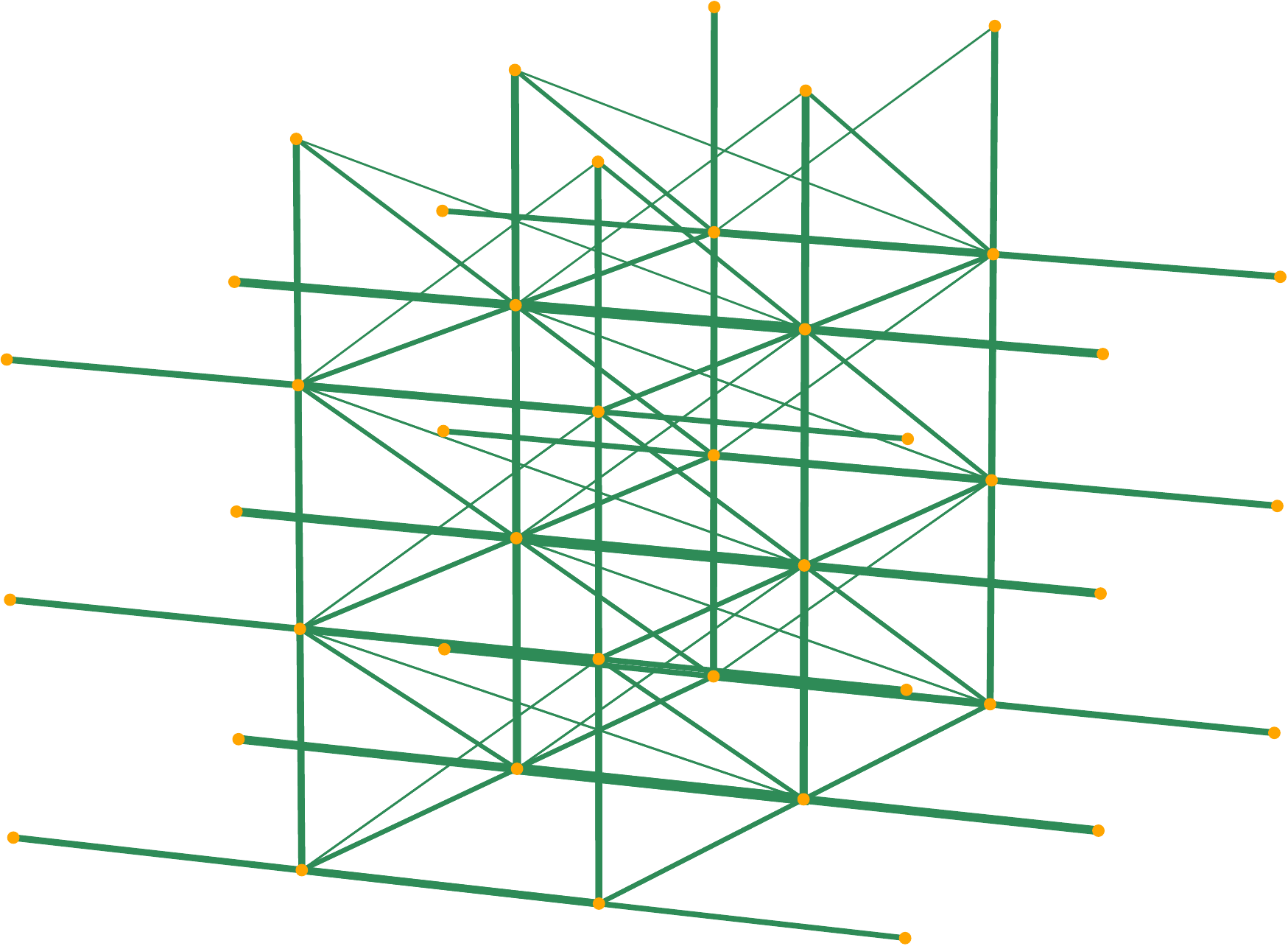}
	\caption{Standard circuit-level depolarization noise decoding lattice for $Z$-stabilizers with distance 3. The edge thickness in the lattice is proportional to $-\ln P$, where $P$ is the overall probability of faults contributing to that edge.}
	\label{fig3}
\end{figure}

A \textit{detection event} occurs when the measured syndrome of a stabilizer differs from that in the previous round, indicating a nontrivial syndrome at the corresponding check node.
By analyzing detection events between connected check nodes, each connecting edge is assigned a weight proportional to $-\ln P$, where the error probability $P$ (as a function of the physical error rate $p$) captures the cumulative effect of noisy stabilizer measurements and error propagation~\cite{wang2011surface, fowler2012topological}.
In general, using a more accurate decoding lattice, which may incorporate additional circuit-level information, can improve decoding performance~\cite{wang2011surface, Fow13}.



$X$-type and $Z$-type errors are handled separately by the MWPM decoder. Given a 3D space-time decoding lattice and the measured syndromes for either $X$-type or $Z$-type errors, a standard MWPM decoder identifies the minimum weight perfect matching in the lattice. This matching corresponds to the most likely error pattern that explains the observed syndromes by connecting them in the decoding lattice.

\subsection{IRMWPM}
It has been shown that iteratively applying MWPM decoding for $X$- and $Z$-type errors with reweighting substantially improves decoding performance~\cite{Fow13,DT14,yuan2022modified, iOMFC23}.
In this section, we describe the IRMWPM decoder with reweighting strategies.
The following subsections detail the reweighting strategies employed in our iterative decoder.

 \begin{algorithm}[htbp]
\caption{Iterative Minimum Weight Perfect Matching (IRMWPM) Decoder}
\label{alg:iterative_mwpm}

\KwIn{ $X$ syndrome $S_X$, $Z$ syndrome $S_Z$, maximum iterations $T_{\max}$ }
\KwOut{X error estimate $\hat{E}_X$, Z error estimate $\hat{E}_Z$}


\tcp{Initialization}
$\hat{E}_X^{(0)} \leftarrow \emptyset$\;
$\hat{E}_Z^{(0)} \leftarrow \emptyset$\;
$k \leftarrow 0$ \hfill\tcp{Iteration counter}
 Construct $X$ decoding graph $\cL_X$ from syndrome $S_X$ \\
Construct $Z$ decoding graph $\cL_Z$ from syndrome $S_Z$ 
\BlankLine

\While{$k < T_{\max}$ }{
    \tcp{X Decoding Step}
   
    \text{Reweight}($\cL_X$,$\hat{E}_Z^{(k)}$)\;
    $\hat{E}_X^{(k+1)} \leftarrow \text{MWPM}(\cL_X)$\;
    
    \BlankLine
    
    \tcp{Z Decoding Step}

    \text{Reweight}($\cL_Z$,$\hat{E}_X^{(k+1)}$)\;
    $\hat{E}_Z^{(k+1)} \leftarrow \text{MWPM}(\cL_Z)$\;
    
    \BlankLine
    
    \tcp{Stopping Criteria}
    \uIf{ $\exists{j,l} \leq k,  \hat{E}_X^{(k+1)} = \hat{E}_X^{(j)} || \hat{E}_Z^{(k+1)} = \hat{E}_Z^{(l)}$,}
    {
        \textbf{break}\hfill \tcp{Error estimates converged}
    }
    
    \BlankLine
    
    $k \leftarrow k + 1$ \;
}

\BlankLine

\Return{$\hat{E}_X^{(k+1)}, \hat{E}_Z^{(k+1)}$}

\end{algorithm}

\subsubsection{Code capacity noise model}
We start with the simpler code capacity noise model, where syndrome extraction circuits are perfect and the decoding lattice constitutes a single horozontal layer,
and then extend the analysis to the circuit-level noise model.

A single-qubit Pauli error can be expressed as $X^a Z^b$ for $a, b \in \{0,1\}$, up to  a global phase.
The joint probability distribution of $X^a$ and $Z^b$ under the depolarizing noise model is given in Table~\ref{tab:joint_probability_single}. In this representation, the two types of errors are not independent. 
In fact, the conditional probabilities $\Pr(X|Z) = \Pr(Z|X) = 1/2$ indicate that the occurrence of one type of error significantly increases the likelihood of the other, raising it from $p/3$ to $1/2$.
However, standard MWPM decoding neglects the correlations between $X$- and $Z$-type errors, making it suboptimal.

\begin{table}[htbp]
	\centering
	\begin{tabular} {|c|p{1cm}|p{1cm}|}
		\hline
		\diagbox{$a$}{$b$} & $0$ & $1$ \\
		\hline
		$0$ & $1-p$ & $p/3$ \\
		$1$ & $p/3$ & $p/3$ \\ 
		\hline
	\end{tabular}
	\caption{Joint probability distribution of $X^a$ and $Z^b$ for single-qubit gates.}
	\label{tab:joint_probability_single}
\end{table}

By incorporating these correlations, the weights of the $Z$-error decoding lattice can be updated based on the decoding results for $X$ errors, and vice versa. This iterative process continues for both $X$ and $Z$ errors until the decoding results converge or a predefined iteration limit is reached.
More precisely, the weight of an edge is determined by the associated conditional probability. When the conditional probability $\Pr(X|Z) = \Pr(Z|X) = 1/2$ is considered, the edge weight is effectively reduced, making it a more favorable matching option.

\begin{figure}[htp]
	\centering
	\subfigure[syndrome]{\includegraphics[width=0.4\columnwidth]{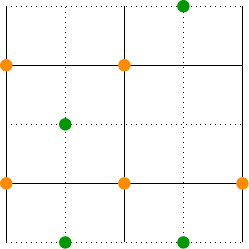}}
	\hfill
	\subfigure[MWPM correction]{\includegraphics[width=0.42\columnwidth]{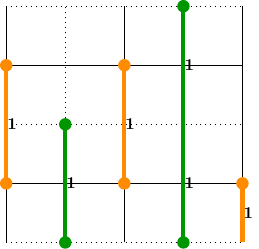}}\\[1ex]
	\subfigure[X matching given Z correction]{\includegraphics[width=0.42\columnwidth]{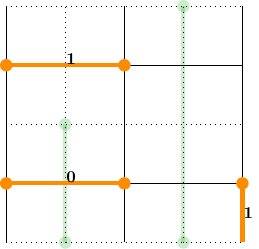}}
	\hfill
	\subfigure[Z matching given X correction]{\includegraphics[width=0.42\columnwidth]{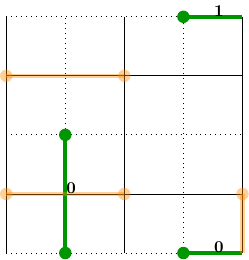}}
	\caption{An example of the iterative weighted strategy on  2D decoding lattices, where green and orange dots represent nontrivial syndromes, and colored edges indicate perfect matchings. }
	\label{fig4}
\end{figure}

For simplicity, the weight of an unconditioned edge is normalized to 1, whereas an edge conditioned on the presence of an error in the dual lattice is reweighted to 0.

Figure~\ref{fig4} provides a simple example illustrating the decoding process of this strategy in a 2D grid when SE is noiseless.
First, the syndromes for $X$- and $Z$-type errors are extracted (Fig.~\ref{fig4}(a)) and decoded separately to obtain their respective initial decoding results (Fig.~\ref{fig4}(b)). The decoding results for $Z$ errors are then used to update the decoding lattice for $X$ errors, leading to refined decoding results for $X$ errors (Fig.\ref{fig4}(c)). Next, these updated $X$-error decoding results are used to adjust the decoding lattice for $Z$ errors, yielding improved decoding results for $Z$ errors (Fig.~\ref{fig4}(d)).

\subsubsection{Circuit-level noise model}

Next, we extend this iterative reweighting decoding procedure to the circuit-level noise model in FTQC, accounting for noise correlations   from error propagation through noisy CNOT gates, which has been studied in~\cite{Fow13,PF23}.
To support reweighting, we propose a systematic method for constructing the conditional probabilities of these correlated events.

Decoding circuit-level noise introduces additional complexities. In the case of noiseless SEs, each edge in the decoding lattice corresponds to a single location fault, making error identification straightforward. However, for circuit-level noise in 3D decoding lattices, a single edge often results from multiple faults across different gates and qubits.

A two-qubit Pauli error generated from two-qubit CNOT gates can be represented as
 $X^{a_1a_2}Z^{b_1b_2}$  for $a_1, a_2,b_1,b_2\in\{0,1\}$, up to a global phase.
The joint probability distribution of $X^{a_1a_2}$ and $Z^{b_1b_2}$ errors for two-qubit gate faults is given in Table~\ref{tab:joint_probability_two}.
Handling $X$- and $Z$-type errors separately can be achieved by considering the distributions of $a_1a_2$ and $b_1b_2$ independently.
For instance, 
\begin{align*}
	\Pr(X^{a_1a_2})=&\sum_{b_1',b_2'=0}^1\Pr(X^{a_1a_2}Z^{b_1'b_2'}  )= \frac{4p}{15},\\
	\Pr(Z^{b_1b_2})=& \sum_{a_1',a_2'=0}^1\Pr(X^{a_1'a_2'}Z^{b_1b_2}  )= \frac{4p}{15}
\end{align*}for $a_1a_2 \neq 00$
and $b_1b_2 \neq 00$.

\begin{table}[htbp]
	\centering
	\begin{tabular}{|c | @{\hspace{2mm}}p{1cm} @{\hspace{2mm}}p{1cm} @{\hspace{2mm}}p{1cm} @{\hspace{2mm}}p{1cm}|}
		\hline
		\diagbox{$a_1a_2$}{$b_1b_2$} & $00$ & $01$ & $10$ & $11$ \\ 
		\hline
		$00$ &  $1-p$ & $p/15$ & $p/15$ & $p/15$\\ 
		$01$ &  $p/15$ & $p/15$ & $p/15$ & $p/15$ \\ 
		$10$ &  $p/15$ & $p/15$ & $p/15$ & $p/15$ \\ 
		$11$ &  $p/15$ & $p/15$ & $p/15$ & $p/15$ \\ [1ex]
		\hline
	\end{tabular}
	\caption{Joint probability distribution  of $X^{a_1a_2}$ and $Z^{b_1b_2}$ errors for CNOT gates.}
	\label{tab:joint_probability_two}
\end{table}

In~\cite{wang2011surface}, single fault mechanisms on surface code decoding lattices, arising from CNOT faults, data qubit faults, and measurement faults, are carefully classified, identifying 23 basic error matchings for each $X$ or $Z$ decoding lattice.
Each error matching is represented by a thick edge connecting two detection events and is categorized into one of six types on the $X$ decoding lattice. These six types, labeled a through f, are illustrated as orange strings in Fig.~\ref{fig5}.
For example, the error matching  a in Fig.~\ref{fig5}~(a)  represents a temporal edge connecting two check nodes corresponding to the same stabilizer across consecutive layers, which  may arise from a measurement fault or a CNOT fault that leaves a residual Pauli error on the measurement qubit.

Also shown in Fig.~\ref{fig5}, each type of orange error matching has corresponding green error matches on the dual lattice that may occur simultaneously  due to the same fault source. 
 For instance, in Fig.~\ref{fig5}~(c), the matchings $\textrm{a}_1$ and $\textrm{a}_2$ correspond to the same type as matching a in Fig.~\ref{fig5}~(a). Every green matching in the dual lattice is correlated with the orange matching in the same subfigure.
For example, the matchings $\textrm{b}$ and $\textrm{d}_1$ in Fig.~\ref{fig5}~(b)
both originate from a single $Y$ error on the data qubit
located at the intersection of the two matchings.
These  six groups of subfigures in Fig.~\ref{fig5}
are adapted from the  classifications in~\cite{wang2011surface}
by pairing matchings that arise from the same fault source.
For example, in the $X$-stabilizer SE in Fig.~\ref{fig2}, a two-qubit error  $X^{01}Z^{10}$ on the control and target qubits after the first CNOT gate
 results in matchings a and d1 in Fig.~\ref{fig5}~(a).
 A similar correlation structure applies to the $Z$ decoding lattice.


In the preceding discussion, each type of error matching in a decoding lattice is associated with a specific occurrence probability. We now examine how the presence of a particular error matching on one decoding lattice may influence its counterpart on the dual lattice.
This reweighting process is guided by conditional probabilities between such correlated edges, effectively transforming a static decoding graph into a dynamic graph that captures the fault correlations imposed by the physical noise model.

To simplify the calculation, we adopt a direct summation approach to estimate the occurrence probability of each type of error matching, under the assumption of a low physical error rate $p$. For instance, the orange error matching a in Fig.~\ref{fig5}~(a) can result from four distinct CNOT gate faults and one measurement fault, corresponding to five out of the 23 error matchings identified in\cite{wang2011surface}. The approximate occurrence probability is given by
\begin{align} 
	P(\textrm{a}) = 4 \cdot \frac{4p}{15} + p = \frac{31p}{15}. \label{eq:Pa} 
\end{align} 
The probabilities for the remaining types of matchings can be computed similarly.

From Fig.~\ref{fig5}, it is evident that matchings on a decoding lattice and its dual can be correlated, and the probability of one matching occurring, conditioned on a correlated matching in the dual lattice, can be computed.
Consider again the orange error matching a in Fig.~\ref{fig5}(a) along with the green error matching d1.
According to~\cite[Fig.~3(a) and~3(b)]{wang2011surface}, these syndromes may arise from errors $X^{10}Z^{10}$ or $X^{01}Z^{10}$ (error matchings 1 and 4), or $X^{11}Z^{01}$ (error matchings 7 and 10).
Thus, the joint probability of occurrence is $P(\textrm{a}, \textrm{d}_1) = 3p/15$.
From Eq.~(\ref{eq:Pa}), the conditional probability is $P(\textrm{d}_1|\textrm{a}) = P(\textrm{a}, \textrm{d}_1)/P(\textrm{a}) = 3/31$, which is significantly higher than the standalone probability of $\textrm{d}_1$, given by $P(\textrm{d}_1) = 42p/15$.
Therefore, we reassign the weight of $\textrm{d}_1$ in the dual lattice as $-\ln(P(\textrm{d}_1|a))$.

Table~\ref{tb:reweighting} lists the conditional probabilities used to reweight edges that are not near the boundaries of the decoding lattices.
For error matchings near the boundary, the corresponding reweighted values may differ due to the reduced number of contributing fault mechanisms.


\begin{table}[htbp]
	\centering
	\renewcommand{\arraystretch}{1.5} 
	\[
	\begin{array}{|c|c|c|c|}
		\hline
		P(\textrm{c}_1|\textrm{a}) = \frac{1}{31} & P(\textrm{d}_1|\textrm{c}) = \frac{3}{16} & P(\textrm{d}_1|\textrm{d}) = \frac{1}{21} & P(\textrm{d}_2|\textrm{e}) = \frac{1}{8} \\
		\hline
		P(\textrm{c}_2|\textrm{a}) = \frac{1}{31} & P(\textrm{d}_2|\textrm{c}) = \frac{3}{16} & P(\textrm{d}_2|\textrm{d}) = \frac{1}{21} & P(\textrm{f}_1|\textrm{e}) = \frac{1}{8} \\
		\hline
		P(\textrm{d}_1|\textrm{a}) = \frac{3}{31} & P(\textrm{e}_1|\textrm{c}) = \frac{1}{8}   & P(\textrm{e}_1|\textrm{d}) = \frac{1}{42} & P(\textrm{f}_2|\textrm{e}) = \frac{1}{8} \\
		\hline
		P(\textrm{d}_2|\textrm{a}) = \frac{3}{31} & P(\textrm{a}_0|\textrm{d}) = \frac{1}{14}  & P(\textrm{e}_2|\textrm{d}) = \frac{1}{42} & P(\textrm{a}_0|\textrm{f}) = \frac{1}{4} \\
		\hline
		P(\textrm{f}_1|\textrm{a}) = \frac{2}{31} & P(\textrm{a}_1|\textrm{d}) = \frac{1}{14}  & P(\textrm{f}_1|\textrm{d}) = \frac{1}{42} & P(\textrm{d}_1|\textrm{f}) = \frac{1}{8} \\
		\hline
		P(\textrm{d}_1|\textrm{b}) = \frac{1}{2}  & P(\textrm{b}_0|\textrm{d}) = \frac{3}{14}  & P(\textrm{f}_2|\textrm{d}) = \frac{1}{42} & P(\textrm{d}_2|\textrm{f}) = \frac{1}{8} \\
		\hline
		P(\textrm{a}_0|\textrm{c}) = \frac{1}{16} & P(\textrm{c}_1|\textrm{d}) = \frac{1}{14}  & P(\textrm{c}_1|\textrm{e}) = \frac{1}{4}  & P(\textrm{e}_1|\textrm{f}) = \frac{1}{8} \\
		\hline
		P(\textrm{a}_1|\textrm{c}) = \frac{1}{16} & P(\textrm{c}_2|\textrm{d}) = \frac{1}{14}  & P(\textrm{d}_1|\textrm{e}) = \frac{1}{8}  & P(\textrm{e}_2|\textrm{f}) = \frac{1}{8} \\
		\hline
	\end{array}
	\]
	\caption{Conditional probabilities used to reweight correlated error matchings across dual decoding lattices.} \label{tb:reweighting}
\end{table}

Note that this correlation between the two types of errors is quite general and not limited to depolarizing circuit-level noise in practice. For example, non-probabilistic noise channels such as amplitude damping ($T_1$ process) also exhibit this correlation from a quantum error correction perspective. In this paper, we focus on circuit-level noise for simplicity. Incorporating these correlations into decoding can also be extended to more realistic noise models in practical scenarios.

\subsubsection{Theoretical guarantees of IRMWPM}
Yuan et al.~\cite{yuan2022modified} showed that under the code-capacity noise model, the  weight of Pauli correction does not increase across IRMWPM iterations. We extend this result to the circuit-level noise model, enabling a proof of convergence for IRMWPM.

\begin{lemma}\label{lemma:nonincreasing}
 In IRMWPM decoding under the circuit-level noise model, the weight of the joint Pauli correction is nonincreasing across iterations. Specifically, let $\hat{E}^{(j)}$ denote the joint Pauli $X$ and $Z$ correction obtained after the $j$-th iteration. Then, for $j\geq 0$,
	\[
	\wt{\hat{E}^{(j)}}\geq \wt{\hat{E}^{(j+0.5)}}\geq \wt{\hat{E}^{(j+1)}}.
	\]
\end{lemma} 
\begin{proof}
 
	Let $\wt{\cL_X^{(j)}, \cM_X^{(k)}}$ denote the weight of the matching $\cM_X^{(k)}$ on the lattice $\cL_X^{(j)}$. 
	These quantities can be expressed as: 
	\begin{align}
		W_{j}=& \wt{\cL_Z^{(0)}, \cM_Z^{(j)}}+\wt{\cL_X^{(j+1)}, \cM_X^{(j)}}\label{eq:w1}  \\
		=& \wt{\cL_X^{(0)}, \cM_X^{(j)}}+ \wt{\cL_Z^{(j)}, \cM_Z^{(j)}}\label{eq:w2}\\
		W_{j+0.5}=& \wt{\cL_Z^{(0)}, \cM_Z^{(j+1)}}+\wt{\cL_X^{(j+1)}, \cM_X^{(j)}}\label{eq:w3}\\
		=& \wt{\cL_X^{(0)}, \cM_X^{(j)}}+ \wt{\cL_Z^{(j+1)}, \cM_Z^{(j+1)}} \label{eq:w4}.
	\end{align}
	Equations~\eqref{eq:w1} and~\eqref{eq:w3} compute the total weight by first evaluating the Pauli $Z$ component, followed by the reweighted Pauli $X$ component. Conversely, Equations~\eqref{eq:w2} and~\eqref{eq:w4} begin with the Pauli $X$ part and then account for the corresponding $Z$ part after reweighting.

We now derive the following sequence of inequalities:
\begin{align*}
	W_{0}=& \wt{\cL_X^{(0)},\cM_X^{(0)}} + \wt{\cL_Z^{(1)},\cM_Z^{(0)}}\\
	\stackrel{(a)}{\geq}&\wt{\cL_X^{(0)},\cM_X^{(0)}} + \wt{\cL_Z^{(1)},\cM_Z^{(1)}}= W_{0.5}\\
		\stackrel{(b)}{=}& \wt{\cL_Z^{(0)},\cM_Z^{(1)}} +\wt{\cL_X^{(1)},\cM_X^{(0)}}  \\
	\stackrel{(c)}{\geq }& \wt{\cL_Z^{(0)},\cM_Z^{(1)}} +\wt{\cL_X^{(1)},\cM_X^{(1)}}=W_{1},
\end{align*}
where  inequality $(a)$  follows from the fact that $\cM_Z^{(1)}$ is a matching of minimum weight on $\cL_Z^{(1)}$,  so $\wt{\cL_Z^{(1)},\cM_Z^{(0)}}\geq  \wt{\cL_Z^{(1)},\cM_Z^{(1)}}$; equality $(b)$ is due to Eq.~(\ref{eq:w3}), which relates the lattices after reweighting; inequality $(c)$ holds because $\mathcal{M}_X^{(1)}$ is a minimum-weight matching on $\mathcal{L}_X^{(1)}$.

By repeating this argument across further iterations, we obtain the monotonic sequence:
\[
W_0\geq W_{0.5}\geq W_{1}\geq W_{1.5}\geq \cdots,
\]
which establishes the claim.
\end{proof}

 We have shown that the sequence ${\wt{\hat{E}^{(j)}}}$ is nonincreasing and bounded below by zero. Therefore, it must stabilize in finite time, establishing the convergence of IRMWPM.

Now we establish the distance guarantee of the IRMWPM decoder.  The standard MWPM decoder provides the following distance guarantee.

\begin{lemma}[MWPM decoder]\label{lemma:MWPM} 
    Let $E = E_X E_Z$ be an $n$-qubit Pauli error on a surface code of distance $d$. 
If $\wt{E_X} \leq \left\lfloor \frac{d-1}{2} \right\rfloor$ and $\wt{E_Z} \leq \left\lfloor \frac{d-1}{2} \right\rfloor$, 
then $E$ is correctable by the MWPM decoder  assuming perfect syndrome measurements.
    \end{lemma}
 
This property is expected to extend to circuit-level noise when more than $d$ rounds of noisy syndrome measurements are used with one additional round of perfect syndrome measurements, although no formal proof currently exists.

The UF decoder offers a similar distance guarantee~\cite{delfosse2021almost};
however, it has been shown that an iterative UF decoder does not maintain this guarantee on surface codes~\cite{LL25}.
This highlights the importance of verifying whether the IRMWPM decoder preserves the same distance guarantee,
which we establish in the following.

\begin{theorem}[Decoding Radius of IRMWPM]\label{thm:itMWPM}
Let $E$ be an $n$-qubit Pauli error with $\wt{E} \leq \left\lfloor \frac{d-1}{2} \right\rfloor$. 
	If $E$ is correctable by the MWPM decoder, then it is also correctable by the IRMWPM decoder on a surface code of distance $d$,
  assuming that the final round of syndrome measurements is perfect.
\end{theorem}

\begin{proof}
 Suppose that $E$ has weight $\mathrm{wt}(E) \leq \lfloor(d-1)/2\rfloor$ and is within the guaranteed decoding radius of MWPM.
By the optimality property of MWPM, the initial estimate $\hat{E}$ satisfies:
$\mathrm{wt}(\hat{E}) \leq \mathrm{wt}(E) \leq \left\lfloor\frac{d-1}{2}\right\rfloor$.
By the non-increasing weight property of IMWPM in Lemma~\ref{lemma:nonincreasing}, after one iteration we have estimate $\hat{E}'$ satisfying:
$\mathrm{wt}(\hat{E}') \leq \mathrm{wt}(\hat{E}) \leq \left\lfloor\frac{d-1}{2}\right\rfloor$.

Since both $\hat{E}$ and $\hat{E}'$ match the syndrome of $E$, the product $\hat{E} \hat{E}'$ has trivial syndrome. Moreover,
$\mathrm{wt}(\hat{E} \hat{E}') \leq \mathrm{wt}(\hat{E}) + \mathrm{wt}(\hat{E}') \leq 2\left\lfloor\frac{d-1}{2}\right\rfloor< d$.
By the definition of the code distance $d$,   $\hat{E} \hat{E}'$ is a stabilizer.
Hence, $\hat{E}$ and $\hat{E}'$ differ by a stabilizer and yield the same logical correction. 

It follows that the IRMWPM decoder produces the same logical outcome as the MWPM decoder when $\mathrm{wt}(E) \leq \left\lfloor \frac{d-1}{2} \right\rfloor$.

 \end{proof}

\section*{Acknowledgement}
 YT and XW were supported by the National Natural Science Foundation of China (Grant No. 92265208).
 CYL was supported by the National Science and Technology Council in Taiwan, under Grants No. 113-2221-E-A49-114-MY3, No. 113-2119-M-A49-008-, and No. 114-2119-M-A49-006-.

\bibliography{quantum}

\end{document}